\documentclass{llncs}

\usepackage{pgf}
\usepackage{tikz}
\usepackage{xcolor}
\usepackage{graphicx}

\usepackage[english]{babel}
\usepackage[OT1]{fontenc}

\newcommand{\pp}{Decide-pp}

\newcommand{\treea}{

\begin{center}
\begin{tikzpicture}[scale=0.8, line width=0.5pt,auto]
 \path
 (2,11)  node[circle,fill=black!30,minimum size=0.1cm](1)  { }
(2,9)  node[circle,fill=black!30,minimum size=0.1cm](2) { }
 (4,8)  node[circle,fill=black!30,minimum size=0.1cm](3) {$2$ }
 (2,7)  node[circle,fill=black!30,minimum size=0.1cm](4) { } 

(2,5)  node[circle,fill=black!30,minimum size=0.1cm](5) { } 
(4,6)  node[circle,fill=black!30,minimum size=0.1cm](4') {$5$ } 

(4,3)  node[circle,fill=black!30,minimum size=0.1cm](6) { }
(6,2)  node[circle,fill=black!30,minimum size=0.1cm](7) {$3$ }
(4,1)  node[circle,fill=black!30,minimum size=0.1cm](8) { }
(6,0)  node[circle,fill=black!30,minimum size=0.1cm](9) {$4$ }
(0,3)  node[circle,fill=black!30,minimum size=0.1cm](10) { }
(0,1)  node[circle,fill=black!30,minimum size=0.1cm](11) { $1$};

\path[-,black] (1) edge[ ] node[left] {$b$}  (2) ;
\path[-,black] (2) edge[] node[left]{} (3);
\path[-,black] (2) edge[] node[left]{$a,c$}(4);
\path[-,black] (4) edge[] node[left]{$\bar{b}$}(5);
\path[-,black] (4) edge[] node[left]{}(4');
\path[-,black] (5) edge[] node[right]{$\bar{c}$}(6);
\path[-,black] (6) edge[] node[left] {} (7) ;
\path[-,black] (6) edge[] node[right] {$e$} (8) ;

\path[-,black] (8) edge[] node[above] {} (9) ;
\path[-,black] (5) edge[] node[left] {$\bar{a},d$} (10) ;
\path[-,black] (10) edge[] node[above] {} (11) ;

 \path
 (12,11)  node[circle,fill=black!30,minimum size=0.1cm](1)  { }
(12,9)  node[circle,fill=black!30,minimum size=0.1cm](2) { }
 (14,8)  node[circle,fill=black!30,minimum size=0.1cm](3) {$2$ }
 (12,7)  node[circle,fill=black!30,minimum size=0.1cm](4) { } 

(12,5)  node[circle,fill=black!30,minimum size=0.1cm](5) { } 
(14,6)  node[circle,fill=black!30,minimum size=0.1cm](4') {$5$ } 

(14,3)  node[circle,fill=black!30,minimum size=0.1cm](6) { }
(16,2)  node[circle,fill=black!30,minimum size=0.1cm](7) {$3$ }
(14,1)  node[circle,fill=black!30,minimum size=0.1cm](8) {}
(16,0)  node[circle,fill=black!30,minimum size=0.1cm](9) {$4$ }
(10,3)  node[circle,fill=black!30,minimum size=0.1cm](10) { }
(10,1)  node[circle,fill=black!30,minimum size=0.1cm](11) { $1$}

(3,0) node[](x){(a)}

(12,0) node[](x){(b)}

(14,11) node[](p){(0,0,0,0,0,0,0,0,0,0)}
(14, 9) node[ ](x) {(0,0,1,0,0,0,0,0,0,0)}
(14,7)  node[] (y) {(1,0,1,0,1,0,0,0,0,0)}
(14,5)  node[](z) {(1,0,1,1,1,0,0,0,0,0) } 
(16,3)  node[](v) {(1,0,1,1, 1,1,0,0,0,0) }
(16,1)  node[](m) {(1,0,1,1, 1,1,0,0,1,0)}
(12,3)  node[](n) {(1,1,1,1, 1,0,1,0,0,0) };

\path[-,black] (1) edge[ ] node[left] {$b$}  (2) ;
\path[-,black] (2) edge[] node[left]{} (3);
\path[-,black] (2) edge[] node[left]{$a,c$}(4);
\path[-,black] (4) edge[] node[left]{$\bar{b}$}(5);
\path[-,black] (4) edge[] node[left]{}(4');
\path[-,black] (5) edge[] node[right]{$\bar{c}$}(6);
\path[-,black] (6) edge[] node[left] {} (7) ;
\path[-,black] (6) edge[] node[right] {$e$} (8) ;

\path[-,black] (8) edge[] node[above] {} (9) ;
\path[-,black] (5) edge[] node[left] {$\bar{a},d$} (10) ;
\path[-,black] (10) edge[] node[above] {} (11) ;

\end{tikzpicture}
\end{center} }

 \newcommand{\forb}{
 \begin{center}
\begin{tikzpicture}[scale=0.7, line width=0.5pt,auto]
 \path
 (2,11)  node[circle,fill=black!30,minimum size=0.1cm](1)  { }
(2,9)  node[circle,fill=black!30,minimum size=0.1cm](2) { }
 (4,8)  node[circle,fill=black!30,minimum size=0.1cm](3) {$s_1$ }
 (1,7)  node[circle,fill=black!30,minimum size=0.1cm](4) { } 
(1,5)  node[circle,fill=black!30,minimum size=0.1cm](5) { $s_2$} 
(3,6)  node[circle,fill=black!30,minimum size=0.1cm](6) {$s_3$ }

 (12,6)  node[circle,fill=black!30,minimum size=0.001cm](s1) { $s_1$} 
(14,6)  node[circle,fill=black!30,minimum size=0.001cm](s3) { $s_3$ } 
(16,6)  node[circle,fill=black!30,minimum size=0.001cm](s2) { $s_2$}

 (13,8)  node[circle,fill=black!30,minimum size=0.001cm](a) { $a$} 
(15,8)  node[circle,fill=black!30,minimum size=0.001cm](b) { $b$ }

 (1,4) node[](v){(a)}

(9,4) node[](u){(b)}

(14,4) node[](u){(c)}

(9, 7.5) node[ ](x) {\, \, $a$ \, $b$}
(9, 7) node[ ](x) {$s_1$ 1 \, 0}
(9,6.5)  node[] (y) {$s_2$ 0 \, 1}
(9,6.0)  node[](z) {$s_3$ 1 \, 1 }

(3,9)  node[](p){(1,0)}
 (2,7) node[](r){(1,1)}
(2,5)  node[](v){(0,1)};

\path[-,black] (1) edge[ ] node[left] {$a$}  (2) ;
\path[-,black] (2) edge[] node[left]{} (3);
\path[-,black] (2) edge[] node[left]{$b$}(4);
\path[-,black] (4) edge[] node[left]{$\bar{a}$}(5);
\path[-,black] (4) edge[] node[left]{}(6);

\path[-,black] (s1) edge[ ] node[ ] { }  (a) ;
\path[-,black] (s3) edge[ ] node[] { }  (a) ;
\path[-,black] (s3) edge[ ] node[ ] { }  (b) ;
\path[-,black] (s2) edge[ ] node[] { }  (b) ;

\end{tikzpicture} 
\end{center}
}

\newcommand{\graph}{
 \begin{center}
\begin{tikzpicture}[scale=0.9, line width=0.5pt,auto]
 \path
 (1, 5)  node[](a)  { $a$}
(2, 5)  node[](b) { $b$}
 (3,5)  node[](c) { $c$}
 (4,5)  node[](d) { $d$}

(3,2) node[](x){(a)}

(8,2) node[](x){(b)}

(14.5,2) node[](x){(c)}

 (14,5.40) node[] {$a$ \, $\bar{a}$\,    $b$ \,  $\bar{b}$  \,  $c$   \, $\bar{c}$ \,   $d$ \, $\bar{d}$ }
 (14, 4.80)  node[] { 1 \,  0  \,   ?  \,   ?  \,   ? \,    ?  \,   ?  \,   ? }
    (14, 4.20) node [] { 1  \,   0  \,   1 \,     0   \,  ? \,    ? \,    ? \,    ? } 
     (14,3.60) node [] { {\bf 1}  \,  {\bf 1} \,    1  \,    0 \,    ? \,    ? \,   1  \,   0} 
       (14, 3.00) node [] { {\bf 1} \, {\bf 1} \,    ?  \,    ?  \,   1 \,   0 \,   1   \,  0} 

(1,3)  node[](1) {1} 
(2,3)  node[](2) {2} 

(3,3)  node[](3) {3}
(4,3)  node[](4) {4}

 (7, 5)  node[](a1)  { $a$}
(8, 5)  node[](b1) { $b$}
 (9,5)  node[](c1) { $c$}
 (10,5)  node[](d1) { $d$} 

(7,3)  node[](11) {1} 
(8,3)  node[](21) {2} 

(9,3)  node[](31) {3}
(10,3)  node[](41) {4};

\path[-,black] (1) edge[]  (a);
\path[-,black] (2) edge[]  (a);
\path[-,black] (2) edge[] (b);
\path[-,black] (3) edge[] (b);
\path[-,black] (3) edge[] (d);
\path[-,black] (4) edge[] (c);
\path[-,black] (4) edge[] (d) ;

\path[-,red] (31) edge[]  (a1);
\path[-,red] (41) edge[]  (a1);
\path[-,black] (21) edge[] (b1);
\path[-,black] (31) edge[] (b1);
\path[-,black] (31) edge[] (d1);
\path[-,black] (41) edge[] (c1);
\path[-,black] (41) edge[] (d1) ;

\end{tikzpicture}
\end{center}}

\newcommand{\grapho} {
 \begin{center}
\begin{tikzpicture}[scale=0.9, line width=0.5pt,auto]
 \path
 (1, 5)  node[](a)  { $c_1$}
(2, 5)  node[](b) { $c_2$}
 (3,5)  node[](c) { $c_3$}
 (4,5)  node[](d) { $c_4$}

(3,2) node[](x){(a)}

(8,2) node[](x){(b)}

(14.5,2) node[](x){(c)}

 (14,5.40) node[] { $c_1$ $\bar{c_1}$ \,$c_2$ $\bar{c_2}$ \,$c_3$ $\bar{c_3}$\, $c_4$ $\bar{c_4}$ }
 (14, 4.80)  node[] { 1 \,  0  \,   1  \,   0  \,   0 \,    0  \,   0  \,   0 }
    (14, 4.20) node [] {1  \,   1  \,   1 \,     0   \,  1 \,    0 \,    1 \,    1 } 
     (14,3.60) node []  {1 \,  1 \,    1 \,    1  \,    1 \,    0 \,   1  \,   0} 
       (14, 3.00) node [] {  1 \,  0 \,    1 \,    1  \,   1 \,   1 \,   1   \,  0} 

(1,3)  node[](1) {1} 
(2,3)  node[](2) {2} 

(3,3)  node[](3) {3}
(4,3)  node[](4) {4}

 (7, 5)  node[](a1)  { $c_1$}
(8, 5)  node[](b1) { $c_2$}
 (9,5)  node[](c1) { $c_3$}
 (10,5)  node[](d1) { $c_4$} 

(7,3)  node[](11) {1} 
(8,3)  node[](21) {2} 

(9,3)  node[](31) {3}
(10,3)  node[](41) {4};

\path[-,black] (1) edge[]  (a);
\path[-,black] (2) edge[]  (c);
\path[-,black] (1) edge[]  (b);
\path[-,black] (2) edge[] (b);
\path[-,black] (3) edge[] (c);
\path[-,black] (3) edge[] (d);
\path[-,black] (4) edge[] (a);
\path[-,black] (4) edge[] (d) ;

\path[-,red] (21) edge[]  (d1);
\path[-,red] (21) edge[] (a1);
\path[-,red] (31) edge[] (b1);
\path[-,red] (31) edge[] (a1);
\path[-,red] (41) edge[] (b1);
\path[-,red] (41) edge[] (c1) ;

\end{tikzpicture}

\end{center}
}

\usepackage{graphicx}
\usepackage{amssymb}
\usepackage{amsmath}
\usepackage{epstopdf}
\usepackage{algorithmic}
\usepackage{a4wide}
\usepackage{verbatim}
\newcounter{ncomm}%

\newcommand{\grb}{$G_{RB}$ }
 \newcommand{\grbc}{$G_{RB}$, }
  \newcommand{\grbd}{$G_{RB}$. }

\newcommand{\m}{$M$}

\newcommand{\me}{{$M_e$} }
\newcommand{\mec}{{$M_e$}, }
\newcommand{\med}{{$ M_e$}. }

\begin{document}
\title{The Binary Perfect Phylogeny with Persistent Characters}
\author{Paola Bonizzoni\inst{1} \and Chiara Braghin \inst{2} \and Riccardo Dondi\inst{3} \and Gabriella Trucco \inst{2}
}

\institute{Dipartimento di Informatica Sistemistica e Comunicazione \\
Univ.  degli Studi di Milano - Bicocca \\
Viale Sarca 336, 20126 Milano - Italy \\
\email{bonizzoni@disco.unimib.it}
\and
 Dipartimento di Tecnologie dell'Informazione
Univ. degli Studi di Milano, Crema \\
\email{chiara.braghin@unimi.it;gabriella.trucco@unimi.it}
\and 
Dipartimento di Scienze dei Linguaggi, della Comunicazione e degli Studi Culturali \\
Univ. degli Studi di Bergamo, Bergamo  \\
\email{riccardo.dondi@unibg.it}.
\email{}}

\pagestyle{plain}
\date{march}\maketitle

\begin{abstract}
The  binary perfect phylogeny model is too restrictive to model  biological events such as back mutations. In this paper we consider a  natural generalization of the  model that allows a special type of back mutation.  We investigate the  problem of reconstructing  a  near perfect phylogeny over a binary set of characters where   characters are {\em persistent}:
 characters  can be  gained and lost at most once. Based on this notion, we define the problem of  the Persistent Perfect Phylogeny (referred as P-PP). 
We restate the P-PP problem  as a special case of the Incomplete Directed Perfect Phylogeny,  called Incomplete  Perfect Phylogeny with Persistent Completion, (refereed as IP-PP),  where the instance is an incomplete binary matrix  $M$ having some missing entries, denoted by symbol $?$, that must be determined (or completed) as $0$ or $1$ so that  $M$ admits a binary perfect phylogeny.
We show that the IP-PP problem can be reduced to a problem over an edge colored graph since the  completion of  each column of the input  matrix  can be  represented by a graph operation. 
  Based on this graph formulation,  we develop  an exact algorithm  for solving  the P-PP problem that is exponential in the number of characters and polynomial in the number of species. 
\end{abstract}

\section{Introduction}

The perfect phylogeny   is one of the most investigated models  in different areas of computational biology.
This model derives from a restriction of the parsimony methods  used to reconstruct the evolution of species (taxa).
Such methods assume that each taxon is characterized by a set of of attributes, called characters. In this paper we focus on the binary perfect phylogeny model; characters can take only the values (or states) zero or one, usually interpreted as the presence or absence of the attribute in the taxa.  Restrictions on the type of  changes from zero to one and vice versa lead to a variety of specific models  (Felsenstein, \cite{Fel}). In the  Dollo parsimony, a character may change state from zero to one only once, but from one to zero multiple times  \cite{Pr1}.  In the variant of Camin-Sokal parsimony \cite{CS},  characters are {\em directed}, only changes from zero to one are admissible on any path from the root to a leaf. This fact means that  the root is assumed to be labeled by the ancestral state with all zero values  for each character, and  no character change back to 0 is allowed. This last variant is  known as the binary directed perfect phylogeny, and it has  a linear time solution \cite{Gus91}. 

Such a  model has been successfully applied  in the context of haplotype inference, starting from the seminal work by Gusfield on the   Perfect Phylogeny Haplotyping Problem \cite{Gus02}.  This last problem has been widely investigated, and very efficient polynomial time solutions have been proposed, including   linear-time algorithms \cite{Gus06}, \cite{Mu}, \cite{Boniz}. However, the real data usually do not fit the simple model of the binary perfect phylogeny and thus in the past years generalizations of the model have been proposed.  Some models are surveyed in \cite{FB00}.

A central goal in this investigation  of the binary perfect phylogeny model is to extend its applicability by taking into account the biological complexity of  data, while retaining  the computational efficiency where possible. More precisely, the binary perfect phylogeny model though allowing a very efficient reconstruction is quite restrictive to explain the evolution of data where   homoplasy events such as back mutations, also called reversals,  are present.  In order to include such events,   the problem of reconstructing the near-perfect phylogeny    has been formalized and investigated.  Some work has been done to produce algorithmic  solutions to the problem, mainly fixed-parameter algorithms have been provided \cite{near-perfect2}, \cite{near-perfect1}.
However,  the near-perfect phylogeny   model appears to be too general for some biological applications.  The model  does not distinguish the main two types of homoplasy occurring in a phylogenetic tree: recurrent mutation  and back mutations.    Back mutations are changes in the  state of the character that only occur along the same path from the root of the tree. On the contrary,  recurrent mutations are changes  in the  state of  the  character  that occur  along different paths of the tree, since the character is allowed to label multiple edges of the tree. In this paper we address the problem of constructing a perfect-phylogeny under the assumption that  only a special type of back mutation  may occur in the tree. A character may change state only twice in the tree,  precisely from $0$ to $1$ and from $1$ to $0$,  and the changes occur along the same path from the root of the tree $T$. These characters have  already been considered in the literature  and called {\em persistent}  by T. Przytycka \cite{Pr06} in a general framework of tree inference.  More precisely,  in  \cite{Pr06},  the change of a character from state $0$ to $1$ models the gain of  the character, while the change from $1$ to $0$ models the loss of the character.  
The use of the notion of persistent character is quite relevant when reconstructing phylogenies that describe the gain and loss of genomic characters \cite{zeng}.  An example of a promising class of genomic characters (also called rare genomic changes - RGC - ) is given by  insertion and deletion  of ‘‘introns’’ in protein-coding genes during the evolution of eukaryotes. In this framework,   persistent characters allow to infer phylogenies by using the gain and loss of introns \cite{zeng}. 

We define a generalization of the (rooted) binary directed perfect phylogeny  where each character may be persistent.   Clearly our model is a restriction of the Dollo parsimony,  where characters can be lost several times, i.e. a character can be lost along different paths from a root to a leaf. Acquisition or loss of characters (i.e. attributes) when
unrestricted could make  the reconstruction of an  evolutionary tree difficult, if not possible.

 Assume that $S = \{s_1, \dots, s_n\}$ is a
set of {\em species} and $C = \{c_1, \dots, c_m\}$ is a set of
{\em characters}. In the paper we consider binary matrices representing
species and characters. More precisely,    a binary matrix
$M$ of size $n \times m$ has columns associated with the set
$C$ of characters, i.e. column $j$ represents character $c_j \in
C$, while rows of $M$ are associated with the set $S$ of
species, i.e. row $i$ represents species $s_i$. Then $M[i,j] = 1$
if and only if species $s_i$ has character $c_j$, otherwise $M[i,j] = 0$. 

In the rest of the paper,  to simplify the notation, we identify rows with species and columns with characters.

The gain of a character in a  phylogenetic tree is usually represented by an edge labeled by the character.   In order to  model  the presence of persistent
characters, the loss of a
character $c$  in the tree is represented  by  an edge that  is labeled by the negation of $c$, or negated character, denoted by $\bar{c}$ .   Clearly, an edge labeled by a negated character  follows  an edge labeled by the  character along a path from the root to a leaf.  
 The following definition is based on the  general  coalescent model given in \cite{EHK} to  represent the evolution of    haplotype sequences and assume that nodes are labeled by vector states of characters.

Formally, we have:

{\bf Persistent Perfect Phylogeny}
Let  $M$ be a binary matrix of size $n \times m$. Then a {\it
persistent perfect phylogeny}, in short  {\em  p-pp tree} for $M$, is a rooted tree $T$
that satisfies the following properties:

\begin{enumerate}

\item each node $x$ of $T$ is labeled by a vector $l_x$ of length $m$;

\item  the root of $T$ is label  by a vector of all zeros,  while  for each  node $x$ of $T$ 
the value $l_x[j]=0, 1$ represents  the state,  $0$ or $1$ respectively, of character $c_j$ in tree $T$;

\item   for each character $c_j$ there are at most  two  edges $e=(x,y)$ and $e'=(u,v)$ 
such that $l_x[j] \neq l_y[j]$ and $l_u[j] \neq l_v[j]$ 
(representing a change in the state of $c_j$) such that $e$, $e'$  occur along the same path 
from the root of $T$ to a leaf of $T$; if $e$ is closer to the root than $e'$, then 
 the  edge $e$ where $c_j$ changes from $0$ to $1$ is labeled $c_j$, 
while  edge $e'$ is labeled $\bar{c_j}$,

\item  each row  of $M$  labels exactly one leaf of $T$.

\end{enumerate}

In the classical  definition of a  Perfect Phylogeny Tree,  in short {\em pp} tree, 
no negated characters are allowed in the tree (see \cite{Setubal}) (see definition in Section \ref{pre1}). Observe that by the above definition  of p-pp tree,  for each  leaf $s$ of tree $T$, the positive characters that label edges 
that are along the unique path from the root to $s$ and
do not occur as  negated along the same path, specify  exactly the characters
that have value $1$ in the row    $s$ of $M$.

Thus, let us state the main problem investigated in the paper.

\noindent

{The \bf  Persistent Perfect Phylogeny
 problem (P-PP):} Given a binary matrix $M$,  returns a p-pp tree for $M$ if such a tree exists.

 \vspace{.2in}
 
 In the paper we investigate the  solution of the P-PP problem. Our main contribution is a graph-based restatement of the problem that allows us to provide an exact algorithm for the problem having a worst time complexity that is polynomial in the number  $n$ of rows of the matrix and exponential   in the number $m$ of characters.  
Since in  real data the number of characters is usually small, while the number of species may be very large, the algorithm could be  efficient even on large instances  as shown by an experimental analysis illustrated in Section \ref{experiment}.
 

 The graph-based solution of   the P-PP problem is obtained by  restating the problem as a  variant of the Incomplete Directed Perfect Phylogeny \cite{Sha}, called Incomplete  Perfect Phylogeny with Persistent Completion (IP-PP), where the input data of this last problem is a specific incomplete matrix  $M$ over values $0,1,?$  and the goal is to complete values $?$ into $0$ or $1$ so that $M$ admits a classical perfect phylogenetic tree.
 Then we show  that the IP-PP problem reduces to the problem of reducing a colored graph by a graph operation  that represents  a completion of a column of the input matrix.    Based on these ideas we discuss our exact algorithm for the P-PP problem.
 
  We believe that the graph-based formulation of the problem could help in investigating polynomial time solutions to the problem.

 \section{The Perfect Phylogeny model: preliminaries}
 
 \label{pre1}

Let us give the definition  of a perfect phylogeny for a binary matrix  and  some relevant basic results that will be used in the paper.

 \noindent
 {\bf Perfect Phylogeny}
 
 Let  $M$ be a binary matrix of size $n \times m$. Then a {\it directed
perfect phylogeny}, in short  {\em  pp tree} for $M$, is a rooted tree $T$
that satisfies the following properties:

\begin{enumerate}
 
 \item each node $x$ of $T$ is labeled by a vector $l_x$ of length $m$;
 
 \item  the root is labeled  by a vector of zeros, while for each  node $x$, the value 
 $l_x[j] = 0,1$ represents  the state,  $0$ or $1$ respectively, of character $c_j$ in tree $T$;

\item  for each character $c_j$ there is at most one edge $e=(u,v)$, labeled $c_j$, such that
$l_u[j] \neq l_v[j]$ (notice that $l_u[j]=0$, while $l_v[j]=1$);
edge $e$ represents a changing of state of $c_j$; 

\item each row of matrix $M$ labels exactly one leaf of $T$.

\end{enumerate}
 
The algorithmic solution of the Perfect Phylogeny model has been investigated in \cite{Gus91}, where a linear time algorithm is provided. 
In particular, the paper \cite{Gus91} provides a well known characterization of matrices admitting a perfect phylogeny.
A binary matrix $M$ admits a perfect phylogeny if and only if it does not contain a pair of columns and three rows inducing    the configurations $(0,1), (1,0)$ and $(1,1)$,  also known as   {\em forbidden matrix} (see Figure \ref{forb} (b)). We will use this  characterization  in the paper.

In particular, the forbidden matrix has a representation by means of a  graph consisting of a path of length four containing three species and two characters; this graph  is   called $\Sigma$-graph. Such a graph is obtained by drawing an edge between every pair of  species and  characters having value $1$ in the matrix (see Figure \ref{forb} (c)). 

Notice that the forbidden matrix is the smallest matrix that does not admit a pp tree. However, by allowing a character to be persistent, the matrix admits a persistent perfect phylogeny, as shown in Figure \ref{forb} (a).

A well known concept that  has been used several times in the framework of the perfect phylogeny is a graph representation of the four configurations $(0,1)$, $(1,1)$, $(1,0)$ and $(0,0)$ (called four gametes): the {\em  conflict graph}. 
We say that two positive characters $c, c'$ of matrix $M$ are in {\em conflict} in matrix  $M$, if and only if   the pair of columns $u,v$ of $M$ induces the four gametes.


\begin{definition}[conflict graph]
\label{sigmag}
Let $M$ be a binary matrix.
The {\it conflict graph} associated with  matrix $M$ is the
undirected graph $G_c =( C,E \subseteq C \times
C)$ where  a pair $(u,v) \in E$ if and only if $u, v$ are in conflict in matrix  $M$.
 \end{definition}


Notice that when  $M$ has  a  conflict graph with no edges,   $M$ does not necessarily admit   a rooted perfect phylogeny,  since $M$ could contain an occurrence of the forbidden matrix. For example, the forbidden matrix has a conflict-graph with  no edges.

In this paper we define a variant  of the the Incomplete Directed Perfect Phylogeny (in short IDP)   \cite{Sha}. The input data of the IPP problem is a matrix over symbols $0, 1, ?$ where symbol $?$ is used to denote  an entry of the matrix that is not determined. Then the IPP problem is finding a completion of the matrix, i.e. assigning values $0, 1$ to $?$ symbols so that the matrix admits a perfect phylogeny.   

For basic  notions of graph theory used in the paper see \cite{Cormen}.

\begin{figure}[ptb]
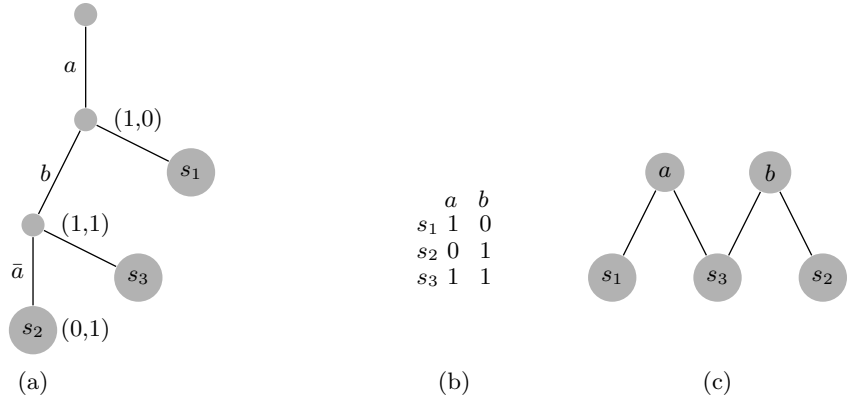

\begin{center}
\forb

\end{center}
\caption{The Figure (a) illustrates the perfect persistent phylogeny for the forbidden matrix reported in Figure (b) and the $\Sigma$-graph for the forbidden matrix in Figure (c). 
} 
\label{forb}
\end{figure}
 
\section{The Incomplete  Perfect Phylogeny with Persistent Completion}

Let $M$ be a binary $n \times m$ matrix which is an
instance of the P-PP problem. The  {\em extended matrix} associated with $M$ is a matrix \me of size ${n \times 2m}$\ over alphabet $\{0,1,?\}$  which is obtained by replacing each column $c$ of $M$ by a pair of columns $(c, \bar{c})$, where $c$ is the positive character, while $\bar{c}$ is the negated character,  moreover for each   row $s$  of $M$, it holds:

\begin{enumerate}

\item    if $M[s,c] = 1$, then  $M_e[s,c] = 1$ and $M_e[s,\bar{c}] = 0$,

\item    if $M[s,c] = 0$, then  $M_e[s,c] = ?$ and $M_e[s,\bar{c}] = ?$.
\end{enumerate}


Informally, the assignment of the pair $(?,?)$ in a species row $s$ for the pair of   entries in columns $c$ and $\bar{c}$ means that character $c$ could be persistent in species $s$, i.e. it is gained and then lost. On the contrary, by definition of a persistent perfect phylogeny,  the pair $(1,0)$ assigned in  a species  row $s$ for the pair of   entries in columns $c$ and $\bar{c}$, means that character $c$ is only gained by the species $s$.

In the paper, we will use  the term extended matrix to denote an extended matrix associated with a binary matrix and defined as above. 
A {\em completion  of  a character $c$}   of matrix \me  is
 obtained by solving the pair $(?,?)$ given in  the pair of columns   $c$, 
$\bar{c}$     by the   value $(0,0)$ or
$(1,1)$.   If a character  $c$ is completed,  then it is  called {\em active}.

A  {\em completion  of matrix}   \me is a completion
of all characters of \mec while a {\em partial completion  of} \me
is a completion of zero  or more  characters of \med

Figure \ref{me} (a) shows an example of  input  matrix $M$ for the P-PP problem. Then Figure \ref{me} (b) shows the incomplete  matrix \me associated with $M$. A possible completion of \me is given in  Figure \ref{me} (c). 

\begin{figure}[htbp]
\begin{center}
\begin{tabular}{c c c c c  c c c c c c c c c c   c c c c c c c c c c  c c c c c c c c c c   c c c c c c c c c c}

$a$ & $b$ & $c$ & $d$ & $e$ & & & & & & & & & & & $a$& $\bar{a}$ & $b$& $\bar{b}$ & $c$& $\bar{c}$ & $d$& $\bar{d}$ & $e$& $\bar{e} $& & & & & & & & & & & $a$ & $\bar{a}$ & $b$ & $\bar{b} $ & $c$& $\bar{c} $ & $d$ & $\bar{d} $& $e$ & $\bar{e}$\\
0 & 0 & 1 & 1 & 0 & & & & & & & & & & & ?&  ?&  ?&  ?& 1&  0& 1&  0& ?&  ?& & & & & & & & & & & 1& 1&  1& 1&  1&  0& 1&  0& 0&  0\\
0 & 1 & 0 & 0 & 0 & & & & & & & & & & & ?&  ?&  1&  0& ?&  ?& ?&  ?& ?&  ?& & & & & & & & & & & 0& 0&  1& 0&  0&  0& 0&  0& 0&  0\\
1 & 0 & 0 & 0 & 0 & & & & & & & & & & & 1&  0&  ?&  ?& ?&  ?& ?&  ?& ?&  ?& & & & & & & & & & & 1& 0&  1& 1&  1&  1& 0&  0& 0&  0\\
1 & 0 & 0 & 0 & 1 & & & & & & & & & & & 1&  0&  ?&  ?& ?&  ?& ?&  ?& 1&  0& & & & & & & & & & & 1& 0&  1& 1&  1&  1& 0&  0& 1&  0\\
1 & 1 & 1 & 0 & 0 & & & & & & & & & & & 1&  0&  1&  0& 1&  0& ?&  ?& ?&  ?& & & & & & & & & & & 1& 0&  1& 0&  1&  0& 0&  0& 0&  0\\

\multicolumn{5}{c}{(a)} & & & & & & & & & & & \multicolumn{10}{c}{(b)}  & & & & & & & & & & & \multicolumn{10}{c}{(c)}\\

\end{tabular}
\end{center}
\caption{The figure  illustrates a binary matrix \m (a) and its  extended matrix \me (b)  and a completion of \me  (c).}
\label{me}
\end{figure}

We introduce below a problem to which we reduce P-PP, as shown in Theorem \ref{equivalence}.


\noindent
{\bf Incomplete  Perfect Phylogeny with Persistent Completion  Problem (IP-PP)}

{\em Instance:} An extended matrix  $M_e$ over   $\{0,1,?\}$.  

{\em Question:}  give a  completion $M'$ of the extended   matrix  $M_e$  such  that $M'$ admits a perfect phylogeny, if it exists.

\vspace{.2in}

Thus we state the first result of the paper. In order to prove the result we assume that the input matrix $M$  does not have columns consisting of zeros, that is for each character $c$, there must exist a species having the character. As a consequence of this assumption, given the extended matrix $M_e$ of $M$, it must be that for each positive character $c$ there is a species having   the positive character  $c$ and not the negated character $\bar{c}$.

\begin{theorem}
\label{equivalence}
Let $M$  be a binary matrix and $M_e$ the extended matrix   associated with $M$. Then $M$  admits a p-pp tree if and only if there
exists a  completion  of $M_e$   admitting
a pp tree.
\end{theorem}
\begin{proof}
({\em If }) Let $M'$ be  a  completion of matrix $M_e$ such that $M'$ admits a  perfect phylogeny  $T$.

In the following we show that apart from the labeling of internal nodes of tree $T$, the tree $T$ is a p-pp  tree for matrix $M$. More precisely, we obtain  a  p-pp tree $T'$ for matrix $M$ by changing the labeling of tree $T$ as follows. 
In order to distinguish the labeling of node $x$ in tree $T$ from the new labeling of the same node $x$ in tree $T'$, we denote the new labeling of $x$ in $T'$ by the vector $l'_x$.
Given 
a node $x$ of tree $T$ labeled by a $2m$ vector $l_x$, then the label  $l'_x$  of node $x$ in tree  $T'$ is defined as follows:

for each $j$, with $1 \leq j \leq m$, $l'_x[j] = 1$ if and only if $l_x[2j -1]=1$ and $l_x[2j]=0$, otherwise  $l'_x[j] = 0.$
Informally, a character $c_j$ has value $1$ in vector $l'_x$ if and only if  $c_j$ occurs as $1$  in vector $l_x$ and it is does not occur as negated in $l_x$, that is  character $\bar{c_j}$ has value $0$ in $l_x$.

We first show that property (4) of the definition of a persistent perfect phylogeny holds for tree $T'$ for matrix $M$, i.e. each row $s$ of $M$ labels a leaf  in tree $T'$. Now, row $s$ of the completion $M'$ labels a leaf  $l_s$ of tree $T$.
We can easily show that $l'_s$ is equal to row $s$ in matrix $M$.   
This fact follows from the observation that characters that have value $1$ in row $s$ of $M$ still have value $1$ in row $s$ of the completion $M'$. By definition of a completion, only   a character having value $0$ in $M$ may be persistent
along a path of tree $T$, i.e. it labels  an  edge of the path and its negated character labels another edge of the same path.
Consequently, the characters associated with the edges along the unique path of $T$ from the root to $s$ and which are not negated are exactly those having value $1$ in row $s$ of $M$, that is $l'_s $ is equal to row $s$ of $M$, as required.

Now,   properties (1) to (3)  given by definition of a persistent perfect phylogeny $T'$ for matrix $M$ follow from the fact that $T$ is a perfect phylogeny for matrix $M'$, and thus each character is associated with exactly one edge of the tree, which implies the same property for each negated character $\bar{c}$.  Observe that by  definition of extended matrix a negated character occurs in a row if and only if its positive character occurs in the same row.  Moreover, since we assume that  for each character $c$ matrix $M$ must have a a species that contains  the positive character $c$, but not the negated character $\bar{c}$, 
 it is immediate to verify that every negated character $\bar{c}$ follows character $c$ along a path from the root to a leaf, thus proving that change of state of $\bar{c}$ from $0$ to $1$ (that is from $1$ to $0$ in tree $T'$)  occurs after the change of state from $0$ to $1$ of character $c$.  In fact, by definition of a completion the set of species  having character $c$ includes the set of species having character $\bar{c}$, since columns $c$ and $\bar{c}$ both have values $1$ or $c$ has value $1$ and $\bar{c}$ has value $0$.

({\em Only if }) Vice versa, let us now show that if there exists a persistent perfect phylogeny $T$ for matrix $M$, then there exists a completion $M'$ of $M_e$ such that $M'$ admits a perfect phylogeny.
We can associate to tree $T$ a matrix $M_{T}$ of size $n \times 2m$ as follows.
  For each  character $c$ of $M$ add a new column $\bar{c}$. Then consider each row $s$ of matrix $M_T$ such that
   a negated character $\bar{c}$ occurs along the path from the root to $s$. Then set  value $1$ for row $s$   in columns $c$  and $\bar{c}$ of matrix $M_T$.
Notice that  $M_T$ is a completion of $M_e$ and clearly $T$ is a perfect phylogeny for $M_T$. 
\qed
\end{proof}



\section{The red-black graph}
\label{section-graph}

In this section  we give a graph representation of  an extended matrix, and we define a graph operation that represents  a special type of completion of  the pair of  columns of the matrix associated with a  character.

Let $M_e$ be an extended matrix $M_e$, then  the {\em red-black} graph $G_{RB}$ for matrix $M_e$ consists of the edge colored graph $(V,E)$ where $V= C \cup S$,  given $C= \{c_1, \cdots, c_m\}$ and $S = \{s_1, \cdots, s_n\}$ the set of positive characters and species of matrix $M_e$, while $E$ is defined as follows:
$(s,c) \in E$ is a black edge  if and only if $M_e[s,c]=1$ and $M_e[s,\bar{c}]=0$.



Then we define a graph operation  on nodes (characters) of the graph \grb that represents a {\em canonical} completion of characters  and consists of adding red edges, removing black or red edges. 
This  graph  operation  over characters (nodes)  of  the red-black graph may be iterated till the graph has only {\em active} characters, as defined below.

\noindent
{\bf Realization of a character $c$  and its canonical completion}

Let   ${\cal C}(c)$ be the connected component of graph \grb containing node $c$.  The {\it realization of   character $c$}  in graph \grb  consists of:

\begin{itemize}

 \item  (a) adding  red edges connecting character $c$ to all species nodes $s$ that are in ${\cal C}(c)$ and such that $(c,s)$ is not an edge of \grbc
  
   \item

  (b) removing  all black edges $(c,s)$ in graph \grbc  Then $c$ is labeled {\em active}.
  
  \item (c) if  a character $c'$ is connected by red edges to all species of ${\cal C}(c)$,  then  $c'$ is called  {\em free}.  Then its outgoing edges are deleted from the graph.  
 
 \end{itemize}
 
 The realization of a character $c$ is associated with a special  completion in matrix $M_e$ of the given character, called canonical.
The {\em canonical} completion of  character $c$ in matrix $M_e$ is defined by completing
each pair $(?,?)$  occurring in the pair of  columns $c$ and $\bar{c}$ as follows:
the pair $(?,?)$  is completed by   $(1,1)$ in  every species $s$ that is in the  component ${\cal C}(c)$ of graph \grbc while value   $(0,0)$ is assigned in the remaining  rows.


\begin{figure}[ptb]
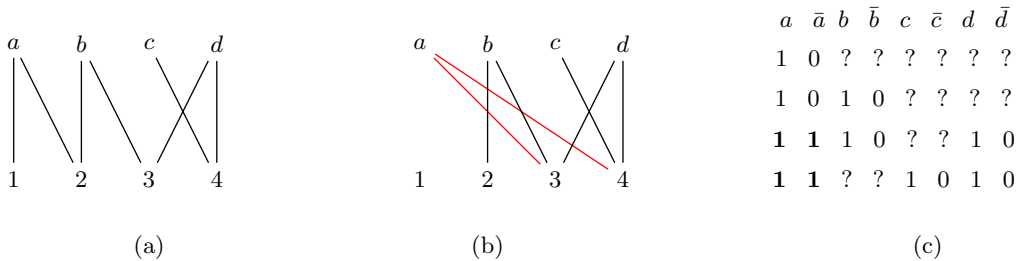



\graph
\caption{The Figure (a) illustrates the  graph \grb of the extended matrix associated with the matrix   of  Example \ref{ex}.  Then  (b)   illustrates the graph  \grb after the realization of character $a$.  Then (c)  illustrates \me after the completion of $a$.}
\label{graph4x4}
\end{figure}

 \begin{example}
 \label{ex}
Figure \ref{graph4x4} (a) illustrates   the red-black graph  of the  extended matrix  associated with matrix $M$ consisting  of rows  $1000, 1100,  0101, 0011$, numbered  $1,2,3$ and $4$, respectively.  Characters of $M$ are denoted by letters $a,b,c$ and $d$. Then Figure \ref{graph4x4} (b) illustrates  the red-black graph obtained     after the realization of   character $a$, while Figure \ref{graph4x4} (c) reports  the corresponding canonical completion in \me of character $a$.

\end{example}

Informally, the red edges of graph \grb  incident to a character $c$ that has been realized  represent  the  pairs  $(?,?)$ in columns  $(c, \bar{c})$   that are completed as $(1,1)$ in matrix $M_e$. 

In the following we call  {\em e-empty} a red-black graph without edges.

\begin{remark}
\label{remark-sigma}

Let $r$ be a sequence of  all characters of a red-black graph for an extended matrix $M_e$, and let  $G_r$  be the graph produced after the realization of the characters in $r$ one after another.  Clearly,  the realization of characters in $r$ produces a completion of the matrix $M_e$.  Then, either $G_r$ is e-empty or $G_r$ contains two nodes inducing a   $\Sigma$-graph. Observe that if  $G_r$ is not e-empty, it must have only red edges  that are  incident to at least two characters of the graph.   In fact,  assume on  the contrary that $G_r$ has a single character  $c$ that  is incident to red edges. Then $c$ must be connected  to all species nodes in the same connected component of the graph. But, this fact leads to a contradiction  since  $c$ is free and all red edges incident to $c$ are deleted from the graph. By  inspection of the possible cases, it is easy to verify that the   minimum size connected component of $G_r$ induces a $\Sigma$-graph consisting of two characters and three species. 
 Such a graph represents the presence of a  forbidden matrix in the  completion $M'$ of matrix $M_e$,  and  hence $M'$ does not admit a  perfect phylogeny (see Section \ref{pre1}). 
 
\end{remark}

The following example illustrates an application of the previous Remark \ref{remark-sigma}.

\begin{example}
\label{exx}
Let  $M$ be  a matrix  having the four characters $c_1$, $c_2$, $c_3$ and $c_4$ and rows $(1, 1, 0,0) $, $ (0, 1, 1, 0)$, $(0,0,1,1)$ and $(1,0,0,1)$, numbered $1, 2, 3$ and $4$, respectively.   Let $G_R$  be the graph obtained after the realization of the sequence $r=<c_1, c_2, c_3, c_4>$  of characters.  Then $G_R$ consists  of  the path $<c_4, 2,c_1,3,c_2,4,c_3>$ with red edges. Then graph $G_R$ induces the $\Sigma$-graph consisting of path  $<2,c_1,3,c_2,4>$.   Observe that the completion $M'$ of  the extended matrix $M_e$ of $M$ consists of rows $(1,0, 1,0, 0,0,0,0)$, $(1,1,1,0,1,0, 1,1)$,  $(1,1,1,1,1,0, 1,0)$ and $(1,0,1,1,1,1, 1,0)$ and such a matrix does not admit a perfect phylogeny as characters $c_1$ and $c_2$ and species $2,3,4$ induce the forbidden matrix in the completion $M'$.
Figure \ref{grapho} illustrates the example.
\end{example}

\begin{figure}[ptb]
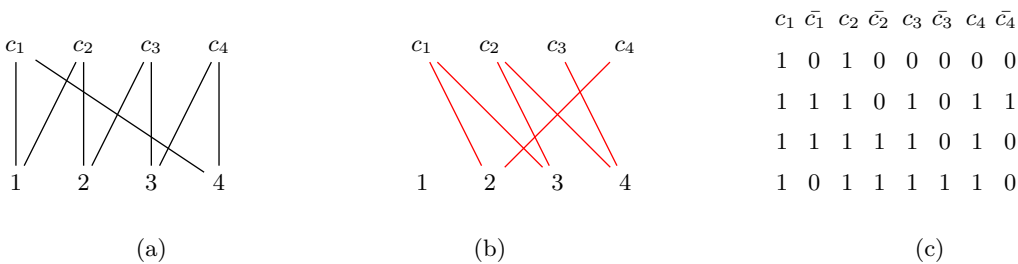



\grapho
\caption{The Figure (a) illustrates the  red-black graph of the Example  \ref{exx} .  Then  (b)   illustrates graph $G_R$ of the example,  while (c)  illustrates   the completion of $M_e$ induced by the realization of sequence $r$.}
\label{grapho}
\end{figure}

Since
 we are interested in computing  canonical completions of the matrix that admit a pp tree   by  the previous Remark \ref{remark-sigma} only canonical completions that are obtained by the realization of  special   sequences of characters of the red-black graph are considered, as defined below.

\begin{definition}

\label{suc-red}
 Given a graph \grb for an extended matrix \mec  a {\em successful  reduction} of \grb is an ordering  $r = <c_{i_1}, \cdots, c_{i_m}>$ of  the set of positive  characters $\{c_1, \cdots, c_m\}$ of the matrix such that the 
consecutive realization of each character in $r$ (which removes
 black edges from graph $G_{RB}$)   leaves an   e-empty red-black graph.

\end{definition}

In   Section \ref{prova}, we show that finding a solution to an instance  of the IP-PP problem   is equivalent to computing the existence of a successful reduction for the red-black graph for the input matrix.
More precisely, let $M_e$ be an instance of the IP-PP problem.  In the following Theorem   \ref{main-reduction},  we prove  that if   \me admits a pp tree $T$, then there exists a successful reduction of graph $G_{RB}$.
Vice-versa,  we show that  a successful reduction of  the red-black graph for  $M_e$ provides a completion $M'$ of the matrix \me that admits a pp tree, thus giving   a  solution to the IP-PP instance..

\begin{theorem} 
\label{main-reduction}   Let \me\ be an extended matrix.
Then  \me\ admits a   perfect phylogeny,  if and only if
 there exists a successful reduction of the graph \grb \ for \med

\end{theorem}

\subsection{Building   a successful reduction from a pp tree}

\label{prova}

This section is devoted to the  proof of  one direction of  Theorem  \ref{main-reduction}, that is showing  how to get a successful reduction from a pp tree.    We  will use the following remark that holds for extended matrices.


In this section, given a node $v$ of tree $T$, by $T(v)$ we denote the subtree of $T$ having root $v$. Moreover, we assume that edges of a tree $T$ are oriented. The orientation of edges is  from the root to a leaf node.


\begin{remark}
\label{annotated0}
Let $T$ be a pp tree for a completion of an extended matrix $M_e$. Then by definition of a perfect phylogeny is immediate to verify that the root  of tree $T$     is  a $0^{2m}$ binary vector of size $2m$, moreover each internal  node  $x$ is labeled by a $2m$-vector $l_x$ defined as follows: each entry $i$ has value $1$ if and only if the  corresponding $i$th  character (positive or negated) occurs along the path from the root of the tree $T$ to node $x$.

\end{remark}

\begin{example}
\label{ex2}

Figure \ref{treea} (a)  illustrates  the pp tree for matrix $M'$ of Figure \ref{me} (c). Notice that rows of matrix $M'$ are numbers $1,2,3,4,5$, while the positive characters are $a,b,c,d,e$. Then Figure \ref{treea} (b) illustrates the vectors  labeling each internal node  of  tree $T$.

\end{example}

\begin{figure}
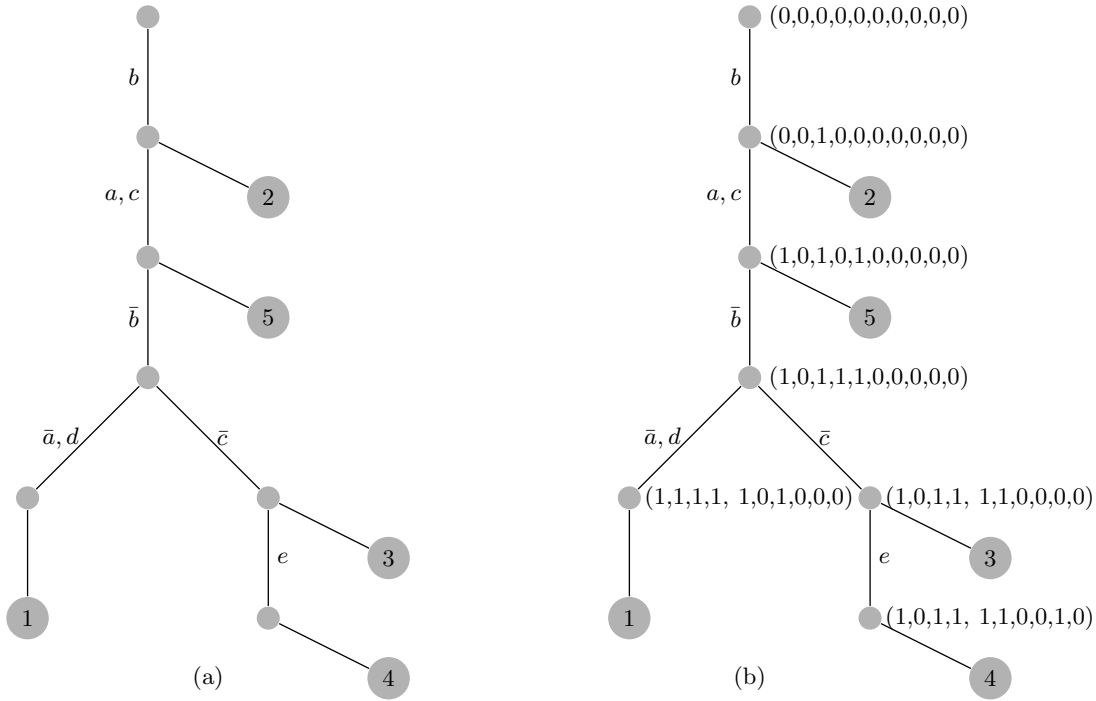


\begin{center}
\treea
\end{center}
\caption{The Figure (a) illustrates a perfect persistent phylogeny $T$ for the matrix of Example \ref{ex2}. Figure (b) reports the vector $l_x$ for each node $x$ of the tree $T$.}
\label{treea}
\end{figure}

Then  we need to state some technical lemmas and introduce a normal form for a pp tree, called {\em standard}. A tree  $T$ is in {\em standard form} when it is in {\em simple form} as defined below and satisfies the properties stated in Definition \ref{standard}.

A pp tree $T$  is in  {\em simple form } if and only if  each edge of the tree is exactly labeled by one character and $T$ does not contain  two  edges $e, e'$ incident to the same node,  one labeled by character $c$ and the other labeled by $\bar{c}$, respectively.  

Given $T$ a pp tree for a completion $M$, we can show that tree $T$ can be reduced to one in simple form. This fact implies that we can obtain from $M$   a  completion $M'$  that admits a tree in simple form. Then the completion  $M'$ is called {\em simple completion}.

\begin{lemma}
\label{tec0}
If $M_e$ has a completion that  admits a pp tree, then  there exists a simple completion $M'$ of $M_e$.  
\end{lemma}

\begin{proof}
Let $T$ be the pp tree  for a completion of $M_e$ and 
assume that $T$ is not in simple form.  We first transform the tree $T$ into a tree $T'$ such that each edge has only one label. Tree $T'$ is   obtained by replacing each edge   with $k >1 $ labels with a path of $k$ edges each one labeled with a distinct label of the replaced edge $e$. On the contrary, an edge without labels is contracted, in the sense that the two end nodes of the edge are collapsed to a unique node. Clearly, the above two operations do not change the completion  $M$.

Assume now that there exists two edges $e, e'$ in  tree $T'$ that are incident to the same node and are labeled by  characters $c$ and $\bar{c}$, respectively. Assume that $e = (x,v)$ and $e'=(v, u)$.  

Then we can move tree $T(u)$ as a subtree of  node $x$ by removing edge $(v,u)$  and making the root $u$ adjacent with node $x$. Then change the completion $M$ by replacing all pairs $(1,1)$ induced by the columns $c$ and $\bar{c}$  and species row in subtree $T(u)$,  by the pair $(0,0)$, obtaining the completion $M'$. It is easy to verify that $T'$ is the tree representation of the completion $M'$.
We can iterate the above operation and obtain a tree in standard form, as required by the lemma.
\qed

\end{proof}

Let $G_{RB}$ be the red-black graph for an extended matrix  $M_e$ and let  $T$ be   the pp tree of a completion of $M_e$.

We require that  a  tree $T$ that is in simple form  satisfies  an additional property  that relates the tree $T$ to the red-black graph  $G_{RB}$.  Observe that we associate to each node $u$ of $T$ the red-black graph that is obtained by the realization of the positive characters that have value $1$ in vector $l_u$ where if  a negated character $c$ is free, then its incident red edges are removed from the graph  if and only if  $c$  has value $1$ in vector $l_u$. We define such a graph, denoted as $G_u$, the {\em red-black graph for node $u$} of $T$.


\begin{definition} [standard property]
\label{standard}
Let $T$ be a pp tree  in simple form for a completion $M_e$. Then $T$ is in standard form if and only if the following property holds:
for each node   $u$ of $T$ such that $(u,x)$ is an edge of the tree $T$ labeled $c$  or $\bar{c}$,    all species of tree $T(x)$ are the same species that are in the  connected component of  graph $G_u$ containing node $c$.

\end{definition}

\begin{lemma}
\label{tec1}
Let  $M_e$ be an extended matrix admitting a pp tree. Then, matrix $M_e$ admits a completion that is represented by a tree $T$ in standard form.
\end{lemma}

\begin{proof}
By Lemma  \ref{tec0} we can assume that the matrix $M_e$ has a simple completion $M_c$ 
that is represented by a tree $T$ in simple form.
To prove the existence of a tree  in standard form, we iterate the application of the following procedure to $T$ till it is in standard form.  Each iteration of the procedure corresponds to changes to the completion $M_c$,
such that the  final computed completion $M'$  is represented by the new tree in standard form. 
Let $u$ be the node that is closest to the root $r$ of $T$, such that $u$ does not satisfy  the property stated in Definition \ref{standard}.  Eventually, $u$ may be the root of $T$.
 Let ${\cal C }(c)$ be the connected component of  the red-black graph  $G_u$ having node $c$.  Since property of   Definition \ref{standard} is violated then the following two cases are possible.
Case 1: there exists a species $s' $ in component ${\cal C }(c)$ and $s'$ is not in subtree $T(x)$.
Case 2:  there exists a set $S'$ of species in the subtree $T(x)$ that are not in ${\cal C }(c)$.

In the following we show that Case 1 leads to a contradiction, while in Case 2 we built  from tree $T$  a new pp tree $T'$ for a completion of matrix $M_e$  where node $u$ does not verify Case 2 and thus $u$ must satisfy Definition \ref{standard}.

 {\em Case 1.}
 
Assume that there exists a species in component ${\cal C }(c)$ that is  not in subtree $T(x)$.
In the following we show that we obtain a contradiction.  If Case 1 holds, then there must exist a species $s'$  that is not in tree $T(x)$  and  is connected in component ${\cal C }(c)$ (by  red or black edges) to a species $s$ of tree $T(x)$  by means of character $c_1$. More precisely, in component ${\cal C }(c)$ $s$ is connected to characters  $c$ and $c_1$, while $s'$ is connected to character $c_1$.  Since $s'$ is not in tree $T(x)$, species $s$ and $s'$ must have a common ancestor that is a node $v$ along the path $\pi_{ru}$ from the root $r$ to node $u$.    
Then  character $c_1$ labels an edge that occur on path $\pi_{ru}$ and thus $v$ occurs before node $u$.  Consequently,   $c_1$ is realized in graph $G_u$, thus implying that $c_1$ is connected  only by red edges to species $s$ and $s'$. This fact implies that both $s, s'$  do not have character $c_1$ and thus it holds that  $\bar{c_1}$ labels an edge that occurs along the path from the root of $T$ to the common ancestor node $v$ of $s$ and $s'$. 
 Since Definition \ref{standard} holds for each node that precedes $u$,  it follows that all species in $T(u)$ are in a connected component ${\cal C}(c')$ that contains $c_1$, where $c_1$ is connected to all species in $T(u)$ by red edges. Consequently,  $c_1$ is free in $G_u$ and thus all red edges connecting $s$ to $s'$ are removed from the red-black graph $G_u$, by definition of red-black graph associated to a node, a contradiction
with the previous assumption. It follows that  species $s'$ cannot be in component ${\cal C }(c)$, thus implying that no other species different from the ones in  $T(x)$ can be in component ${\cal C }(c)$.
 
 {\em Case 2.}
 
 Assume that $S'$ is the largest set of species that are in $T(x)$ and are not in the component 
${\cal C }(c)$. 
By definition of pp tree, it must be that the set of species $S'$ is in the subtree $T(v)$ for a node $v$ of degree at least $2$ that is along the path from node $x$ to a leaf and is the end node of  the edge $(z,v)$  labeled $\bar{c}$. In fact, since $S'$ contains species that are not in  ${\cal C}(c)$, by definition of graph $G_{RB}$ it means that such species do not have character $c$ and consequently they must occur in the subtree  induced by the end node of the edge labeled $\bar{c}$.

Now, consider the path $\pi_{xv}$ from node $x$ to the node $v$. Let $y$ be the node on path $\pi_{xv}$ such that given the unique path $\pi_{yv}$ from node $y$ to $v$, it consists of only degree $2$ nodes. If such $y$ does not exist, then pose $y=z$.

Let $T_y$ be  the subtree  of $T$ consisting of path $\pi_{yv}$ and subtree $T(v)$. Clearly, all species $s'$ of $T_y$ are the ones of subtree $T(v)$, by construction of $T_y$.

Consider subtree $T'_y$ which is obtained from subtree $T_y$ after removing the edge labeled ${\bar c}$ (see Figure \ref{lemfi}).
In the following we show that subtree $T'_y$ can be moved as a subtree of node $u$  thus obtaining a new tree $T'$ such that is a pp tree for a completion of matrix $M_e$. By construction,  in tree $T'$   subtree $T(x)$ does not contain the set $S'$ of species that  are not in component ${\cal C} (c)$ and thus Case 2 does not hold for node $u$ in tree $T'$.
 Moreover, by application of the constructive procedure given in the proof of  Lemma \ref{tec0}, tree $T'$ can be reduced to a simple form, thus proving what required.

Assume to the contrary that tree $T'_y$   cannot be moved as subtree of node $u$ to obtain a pp tree for a completion of matrix $M_e$. Observe that a species $s$ in subtree $T_y$  of $T$ has all positive characters of the path  from the root $r$ to node $x$  that are not negated along the path $\pi_{xy}$.  Moreover $s$ has positive characters that occur in $\pi_{xy}$. Thus, tree $T'_y$  cannot be moved as subtree of node $u$ in $T'$,  if   two cases hold: (i)  tree $T'_y$ has a species $s$ containing a  positive character  $e$  that belongs to path ${\pi}_{xy}$ or (ii) $s$  does not have a character $e$ that occurs  as negated in path ${\pi}_{xy}$ and occurs as positive in the path from the root $r$ to node $x$.
Let us consider case (i). Thus assume first that tree $T_y$ has a species $s$ having character $e$  that is on the path ${\pi}_{xy}$.  Since $y$ has degree bigger than $2$, there exists a species $s'$ such that $e$ is a character of $s'$ and $s'$ occurs at the end of  a path leaving  node $y$ that is distinct from the path $\pi_{yv}$ from node $y$ to node $v$.  Moreover, $c$ must be a character in $s'$, since the character $\bar{c}$ occurs after node $y$ along path $\pi_{yv}$.
Consequently, in the red-black graph $G_{RB}$ species $s'$ is connected to characters $c$ and $e$, where $e$ is connected to character $s$. It follows that $s$ is connected to $c$ by means of character $e$.
This situation of the red-black graph $G_{RB}$ is also present in the graph 
$G_u$, as both $e, c$ are not   realized in graph $G_u$, being $c, e$ labels of edges that occur after node $u$.  Consequently,  $s$ is in the   connected component ${\cal C}(c)$, thus obtaining  a contradiction with the fact that $s \in S'$.
Let us consider case (ii). 
Thus assume now that tree $T_y$ has a species $s$ that does not have character $e$, where $e$ occurs as positive on the path from the root $r$ to node $x$, while  character $\bar{e}$   occurs on the path ${\pi}_{xy}$. 
Similarly as above, since $y$ has degree bigger than $2$, there exists a species $s'$ occurring at the end of a path leaving node $y$  that is distinct from path $\pi_{yv}$ and $s'$ contains characters $\bar{e}$ and character $c$. Clearly, $e$ labels an edge $(l_1, v)$ that occurs  before node $u$ and since the standard property holds for each node above $u$,  it follows that $s'$ and $s$ are in the same connected component of graph $G_{l_1}$    having character $e$ and character $c$. Thus, when character $e$ is realized in graph $G_{l_1}$, both species $s, s'$ are connected to $e$ by  red-edges, thus implying that $s$ is in the same connected component of  node $c$.  Observe that  this property holds also for graph $G_u$. In fact, $\bar{e}$ occurs after node $u$, and thus red-edges incident to node $e$  cannot be removed from   graph $G_u$.  Consequently, $s$ is in the connected component ${\cal C}(c)$, thus    contradicting the assumption that $s \in S'$. Since case (i) and (ii) are not possible, it follows that all species in subtree $T'_y$ do not have (positive or negated) characters of path  ${\pi}_{xy}$ in the completion of matrix $M_e$. It follows that tree $T'_y$ can be moved as a subtree of node $u$ to obtain tree $T'$ where $T'$ is a pp tree for the completion $M'$  of matrix $M_e$ obtained by replacing in all rows $S'$ of the completion $M_c$ the   entry  $(1,1)$ in columns $c$ and $\bar{c}$ by the pair $(0,0)$.  In fact, species in $S'$ will not have  the character $c$ and its negated character in the new tree $T'$ and thus $T'$ is the pp tree for the new completion $M'$. This observation completes the proof of Case 2.  \qed

\end{proof}

\begin{figure}[ptb]
\begin{center}
\includegraphics[width=0.55\textwidth]{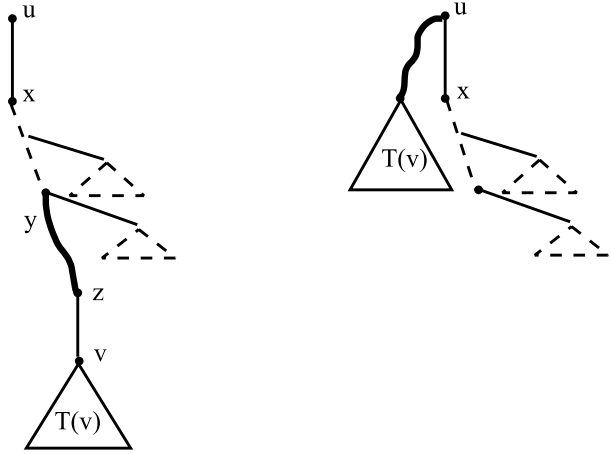}
\end{center}
\caption{The figure illustrates Case 2 of  Lemma \ref{tec1}, more precisely it represents the operation of  moving tree $T'_y$  as subtree of node $u$. Observe that tree $T'_y$ consists of subtree $T(v)$  and the path $\pi_{yv}$ where edge $(z,v)$ labeled $\bar{c}$ has been removed.}
\label{lemfi}
\end{figure}

We show the first preliminary lemma that allows us to prove the main Theorem of the paper.

\begin{lemma}
\label{prelim-annotated}
Let $T$ be a tree in standard form representing the completion $M'$ of an extended matrix $M_e$. Let $G_x$ be the red-black graph for a node $x$ of $T$.
 Then, given $C'$ the set of positive characters having value $1$ in vector $l_x$, the completion of $C'$ in $M'$ is canonical, i.e. it is given by the realization of characters $C'$ in graph   $G_{RB}$.   Moreover,  a negated character  $c$ changes value from $0$ to $1$ in $l_x$ if and only if it is free in graph $G_x$. Then edges incident to  $c$ are removed from $G_x$.

\end{lemma}

\begin{proof}
Let $G_{RB}$ the red-black graph for the extended matrix $M_e$.
We show the lemma by induction on the number $k$ of ones occurring in a node $l_x$ of tree $T$. 
Assume first that $k =1$ and $c$ is the character that has value $1$ in vector $l_x$. Clearly, $c$ is the  label of the edge $(r,x)$ where $r$ is the root of tree $T$ and $c$ is a positive character. Observe that the completion of columns $c$ and $\bar{c}$ in $M'$ is a canonical completion, i.e.  it is obtained by the realization of node $c$ in graph $G_{RB}$.   In fact,  since by  Definition \ref{standard} of a standard tree   the  species of tree $T(x)$  are  the same ones of the  connected component of graph \grb  having  character $c$, it holds that they must have value $(1,1)$  or $(1,0)$ in  the pair of columns $(c, \bar{c})$ of matrix $M'$. 

Now, assume that the number of entries with value $1$ in  vector $l_x$ is $k >1$ and the edge $(u,x)$  is incident to node $x$ (assuming that edges are oriented following every path from the root $r$) .
Given vector $l_u$, since the number of  entries that are $1$ in $l_u$ is less than the number of  entries that are $1$ in $l_x$,  by induction, the completion of all characters $C'$  that are in $l_u$ is given by the realization of the set $C'$ in graph $G_{RB}$, thus obtaining  the red-black graph $G_u$  for node $u$.   
Since $T$ is in a standard form,   two cases are possible: (1) edge $(u,x)$ is labeled by a positive character $c$ or   (2) a negated character $\bar{c}$. 

Let us consider case (1) first.
Clearly, character $c$ is non active in $G_u$ as it has $0$ value in vector $l_u$.
 By Definition  \ref{standard} of a standard tree all species in $T(x)$ are the same ones that are in the  connected component of graph $G_u$ having character $c$.  Notice that    species in $T(x)$  specify    the rows where column $c$ and $\bar{c}$ must have the value $(1,1)$ or $(1,0)$ in the  completion $M'$. Consequently, the completion of $c$ in matrix $M'$ corresponds to the realization of $c$ in the red-black graph $G_u$. Thus the completion in matrix $M'$ of  the set $C' \cup \{c\}$ of  characters in $l_x$ is given by the realization of such characters in graph $G_{RB}$, thus proving that the lemma holds in this case.

Let us consider case (2).
Assume that edge $(u,x)$ is labeled by the negated character $\bar{c}$. Then by induction the completion of all columns for positive characters in $l_u$ is given by their realization in graph $G_{RB}$. Since edge $(u,x)$ is labeled by a negated character, it follows that the red-black graph $G_x$ is obtained by the realization of  positive characters in $l_u$.  Consequently, $G_u$ and $G_x$ are the same graph, thus showing that the completion of  positive characters in $l_x$ is   given by the realization of positive characters in $G_x$. Since $\bar{c}$ has value $1$ in node $l_x$ and value $0$ in node $u$, we must show that  $c$ is  free in graph $G_u$. In fact,  all species in $T(x)$ do not have character $c$, and by Lemma \ref{tec1} these species  are exactly the species that are  in the connected component of node $c$ in graph $G_u$. This fact implies that  character $c$ is free in graph $G_u$. 

Now, let us show that   character $c$ cannot be free   in graph $G_v$, for $v$ a node that is an ancestor of $u$. In fact, since tree $T$ is in standard form, it must be that all species in subtree $T(v)$ are in the same connected component of a character $c_1$ where $c_1$ labels an edge  $(v, z)$. But, since $\bar{c}$ occurs after node $z$, it follows that there exists a species $s$  in subtree $T(v)$ that has the positive  character $c$ and  therefore  no  red edge incident to $s$ a node $c$ is given in graph $G_v$. This fact implies that $c$ is not free in graph $G_v$.

This fact  completes  the proof of the lemma.
\qed


\end{proof}

By  Lemma \ref{prelim-annotated}, the completion of characters in an extended matrix $M_e$ that admits a tree in standard form is a canonical one.

\begin{corollary}
\label{cor-annotated}
Let $T$ be a tree in standard form  representing the completion $M'$ of an extended matrix $M_e$. Then $M'$ is a canonical completion.
\end{corollary}

\begin{proof}
The result is a direct consequence of the previous Lemma \ref{prelim-annotated}. In fact, given a node $x$ of tree $T$, by induction on the number of $0$ that  are in   vector $l_x$,  it is easy to show by  direct application of Lemma \ref{prelim-annotated} that the completion of all characters in tree $T(x)$ is given by   their realization   in the red-black graph $G_x$ for node $x$. This fact implies that the completion of all characters in tree $T$ is canonical.
\qed
\end{proof}

 In the following we can show that a pp tree $T$  represents a successful reduction of the red-black graph.
 
\begin{lemma}
\label{first-direction}
Let $G_{RB}$ be the red-black graph for an extended matrix  $M_e$. If $M_e$  admits a  pp tree, then there exists a successful reduction of graph $G_{RB}$.
\end{lemma}

\begin{proof}

By Lemma \ref{tec1},   there is a  completion $M$ for $M_e$ that  admits  a tree $T$ in standard form.
Let $G_x$ be the red-black graph for  node $x$ of $T$. 
  Then we prove that there exists a successful reduction of $G_x$, by 
 induction on the number $k \geq 0 $ of  $0$ entries that are left in the root vector   of tree $T(x)$.
Assume first that $k =0$, i.e. all entries of the root have value $1$.  This fact implies that all characters have been realized in the red-black graph. By construction, the red-black graph can only have red edges.  By Remark \ref{remark-sigma},  if it is not e-empty it means that it has a $\Sigma$-graph. But, this fact implies  a contradiction with the existence of the tree $T$. In fact, the $\Sigma$-graph represents the existence in $M$ of the induced forbidden matrix. Assume that $a$ and $ b$ are the two characters of the forbidden matrix. Since by Corollary \ref{cor-annotated}, the completion of columns for $a$ and $b$ in $M$ is canonical, i.e. is associated with the realization of $a$ and $b$ in graph $G_{RB}$,  we obtain a contradiction. Thus $G_x$ must reduce to the e-empty graph, i.e. a successful reduction of $G_x$ must exist.

Assume now that the number of entries $0$ in vector $l_x$ is equal to $k$, with $k =m $ and $m >0$. Then the root  $x$ of tree $T(x) $ has an outgoing edge $(x,u)$ that is labeled by a character $c$ which means that the entry of $c$ in vector $l_u$ is $1$, while it is $0$ in $l_x$. Two distinct cases must be considered (1) either $c$ is positive or (2) $c$ is negated.

Case 1.  Assume $c$ is positive.
Since  vector $l_u$ has one zero less than the root of tree $T$, that is $k = m-1$. By  induction  it holds that the red-black graph $G_u$ reduces to the e-empty graph. Since graph $G_x$ differs from graph $G_u$ by the realization of $c$, it follows  that $G_x$ reduces to the e-empty graph by the realization of $c$ and non active characters in $G_u$.

Case 2. If $c$ is negated, by Lemma \ref{prelim-annotated}  the red-black graph $G_x$ for node $x$ is equal to the  red-black graph $G_u$ and $c$ is free in $G_u$.  Since vector $l_u$ has a $0$ entry less than  the number of $0$ entries in  $l_x$, by induction $G_u$ reduces to the e-empty graph and consequently also $G_x$ reduces to the e-empty graph. Thus, both cases prove that there exists a successful reduction of  graph $G_x$.
\qed

\end{proof}

The previous Lemma \ref{first-direction} provides the proof of the {\em Only if} direction of Theorem \ref{main-reduction}.

\subsection{Building a pp tree from a successful reduction}

In this section we complete the proof of Theorem \ref{main-reduction}  by showing that  a successful reduction provides a completion for a matrix $Me$ admitting a pp tree.

\begin{theorem}
Let $M_e$ an extended matrix. If there exists a successful reduction of the graph $G_{RB}$, then $M_e$ admits a perfect phylogeny.
\end{theorem}

\begin{proof}

Let   $M$ be the completion of matrix \me\ obtained from a successful reduction of the red-black graph for $M_e$.  
In the following we show that $M$ has no forbidden matrix. This fact will prove that $M$ admits a pp tree.
Let $G_R$ be the red-black obtained after the realization of the characters of the successful reduction.
 Assume to the contrary that $M$ has two characters $c, c_1$ that induce a forbidden matrix $F$, and let $s_1, s_2, s_3$ be the species of $M$  having the configuration $(1,1)$, $(1,0)$ and $(0,1)$ in $F$, respectively.
 
We must consider the following cases.

Case 1. Assume that the forbidden matrix is induced by   two negated characters.
This fact implies that $G_R$ will have an induced $\Sigma$-graph, thus contradicting the fact that $G_R$ is e-empty.

Case 2. Assume that the forbidden matrix is induced by two positive characters.

 Then  $c, c_1$ must be in the same connected component of the  red-black graph before their realization, as species  $s_1$ is connected to both characters (we do not know if $s_1$ is connected  by a black or red edge).
 Now, the realization of $c$ produces the red edge $(c,s_3)$,  since $M_e[s_3,c] = ?$. Then $M[s_3,c]=1$ in the completion  $M$, a contradiction with the assumption.
 
 Case 3. Assume that the forbidden matrix is induced by a positive and negated character.

Assume that $c$ is the negated character.  Since $c$ and $c_1$ share a species in the forbidden matrix $F$, it means that $c$ and $c_1$ are in the same connected component of the red-black graph  when $c_1$  and $c$ are realized in the graph. Since $(0,1)$ is given in the matrix $F$ in row $s_2$, by definition of realization of $c_1$, it must be that  $M[s_2, c_1]=1$ and $M[s_2,\bar{c_1}]= 1$ as $M_e[s_2,c_1]=0$. But this  is a contradiction.

\qed
\end{proof}

\begin{figure}[ptb]
\begin{center}
\includegraphics[width=0.6\textwidth]{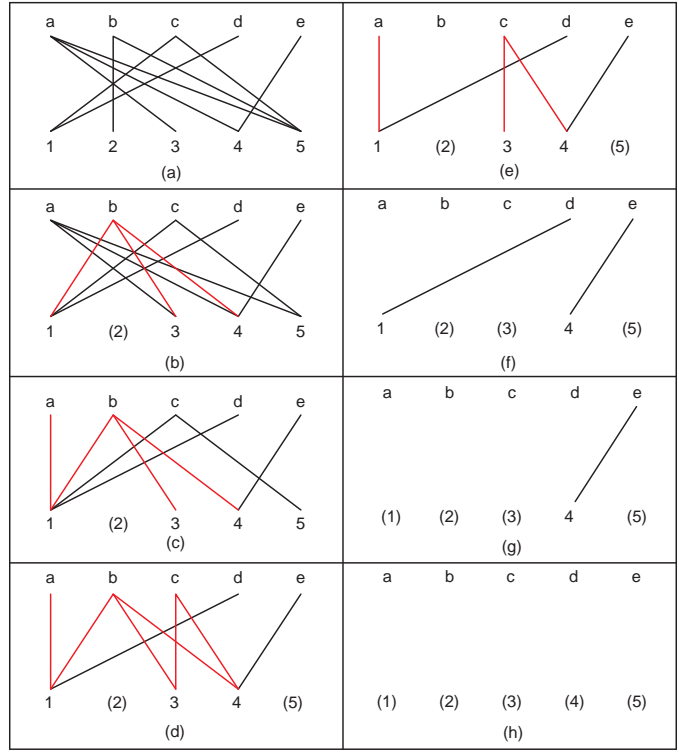}
\end{center}
\caption{The main steps of the successful realization of  graph \grb   described in Example \ref{example-reduce-graph}.}
\label{graph}
\end{figure}

\begin{example}
\label{example-reduce-graph}
Let us consider the $5 \times 5$ input matrix \me shown in Figure \ref{me} (b). In the following we detail the realization of characters of a successful reduction   $r = < b, a, c, d, e>$ of graph \grbd
First of all, observe that  Figure \ref{graph} (a) illustrates the initial red-black graph \grbd In the following we say that a species is realized when it is a singleton in the red-black graph.

First character $b$ is realized (Figure \ref{graph}(b)) and then  
the species $2$ is realized. Then character $a$ is realized (Figure \ref{graph}(c)). 
Note that we do not add any edge incident to species $2$, since it has been already realized. 
Then character $c$ is realized (Figure \ref{graph}(d)). 
and  species $5$ is realized. Moreover, character $b$ is free since it is connected by red edges to all species of the same connected component of $b$ (Figure \ref{graph}(e)).  Since character $a$ is connected to all species of its connected component with  red edges, $a$ is free.
The same fact holds for character $c$. At this point, species $3$ is a singleton, so it is realized (Figure \ref{graph}(f)).
Then character $d$ is realized (Figure \ref{graph}(g))  and species $1$ is realized. 
 Finally,   character $e$ is realized  and  so species $4$ is realized. At this point (Figure \ref{graph}(h)), \grb does not contain any edge.

Notice that the successful reduction provides the completion that is given in Figure \ref{me} (c).

\end{example}

Observe  that a successful reduction of the red-black graph provides the main steps of the process of building  a pp tree. More precisely,  the realization of a single character leads to an operation in the pp tree, which is either adding an edge labeled by a character or adding a leaf node corresponding to a species node.



\section{An exact algorithm for the P-PP problem}
In this section we propose an algorithm for   the P-PP problem that is based on Theorems \ref{equivalence} and \ref{main-reduction}.
The algorithm   reduces an instance $M$ of P-PP to an instance $M_e$ of the  IP-PP problem.  By the proof of Theorem \ref{equivalence} $M_e$  admits a pp tree $T$ if and only if $T$ is a solution of  matrix $M$. 
Then by the characterization given by Theorem \ref{main-reduction},  $M_e$ admits a pp tree $T$ if and only if there exists a successful reduction of the red-black graph for $M_e$.   We design an algorithm, called {\bf  \pp}  that  builds a  decision tree that explores all permutations of  the set $C$ of  characters  of $M_e$ in order to find one that is a successful reduction, if it exists.  More precisely, each edge of the decision tree  represents a character and each path of the tree from the root to a leaf is a distinct permutation of the set $C$.  
The algorithm works in a branch and bound like manner, in the sense that if a branch of the decision tree ending in node $x$ does not lead to a solution, then the decision tree below $x$ is discarded.  More precisely,  each branch ending in node $x$   gives a partial permutation  $\pi$  that consists of all characters labeling the path from root $r$ to  node $x$. A partial completion  $M_{\pi}$ is computed by realizing characters provided by the partial permutation $\pi$. Whenever $M_{\pi}$ contains the forbidden matrix,  then the branch ending in $x$ does  not lead to a solution, and $x$ is labeled as a {\em fail} node.

 
\vspace{.2in}

Below we give a general procedure for the realization of a single character in the red-black graph built during the realization of characters.

{\bf Procedure Realize($c, M',$\grb )}

 {\em Input:} a  character $c$, a partial completion $M'$ and a red-black graph \grb  

 {\em Output:} character $c$ is realized in graph \grb and $c$ is completed in  $M'$

\begin{itemize}

\item Step 1.  Mark character $c$ as {\em active}.
\item Step 2.  Compute the connected component $\cal C$ of graph \grb containing character $c$
\item Step 3.  Realize character $c$:

- add red edges connecting character $c$ to all  species nodes $s'$  in $\cal C$ such that $c$ is not connected to $s'$,

	- remove all black edges $(c, s)$ in $\cal C$,
	
	- update the graph \grb by removing all red-edges outgoing  from a character $c'$ of \grb that is {\em free}.

\item Step 4.  Complete the  columns of characters $c$ and $\bar{c}$ in $M'$   as follows:   in every row $s$ such that $(c,s)$ is a red edge in \grbc,  replace   the  pairs $(?,?)$ by $(1,1)$,  otherwise     by $(0,0)$.


\end{itemize}



Let us now describe the main algorithm that consists of 
 {\bf \pp ($M_e, r, \{r\}$)} call, where $r$ is the root of the decision tree,  and initially the visited tree consists of set $\{r\}$.



{\bf Algorithm \pp($M', x,T$)}

 {\em Input:}  a partial depth-first visit tree $T$ of the decision tree $\cal T$ and a leaf  node $x$ of ${\cal T}$,  a partial completion $M'$  obtained by the realization of  the  characters labeling a path $\pi$ from $r$ to node $x$

 {\em Output:} the tree $T$ extended with the depth-first visit of $\cal T$  from  node $x$.
The procedure eventually outputs a successful reduction or a complete visit of $\cal T$ that fails to find such a successful reduction.

\begin{itemize}

\item Step 1.  if  the edge incident to node $x$ is labeled $c$, then  {\bf Realize($c, M',$ \grb )}. 

 If the matrix $M'$ has a  forbidden matrix, then label $x$ as a {\em fail} node. Otherwise, if $x$ is a leaf node, then mark $x$ as a  successful node and output the permutation labeling the path from the root $r$ of tree $T$ and leaf node $x$.

\item Step 2. For each node $x_i $  that is a  child of node $x$ in tree $\cal T$ and is labeled by a non active character, 
apply  {\bf \pp}($M', x_i, T \cup \{x_i\}$).

\end{itemize}

\subsection{Complexity}
The worst case time  of the algorithm is achieved when  the whole permutation tree $\cal T$ is visited. Generating all permutations requires $m!$ time.
Each time a character $c$ is realized  all species of the matrix are examined in the worst scenario.   Moreover, the connected components of the red-black graph must be updated each time.
Thus, the realization of  $m$ characters has a time complexity that is $O(n \cdot m) \times O(g(n,m))$, where $g(n,m)$ is the cost of maintaining connected components of a graph whose size is $O (n^2 \times m)$.  Since red edges are added  to the graph, in the worst  scenario each species   will have $O(n)$ incident red-edges. 

A trivial implementation of the connected component update  would require $\alpha(n^2 \cdot m )$ each time a character is realized, $\alpha$  being  the inverse of the Ackerman function. More efficient  implementations can be obtained  by  using dynamic algorithms \cite{Holm}.   Thus  totally,  the time is $O(m! \cdot n \cdot m \cdot \alpha(n^2 \cdot m))$.
This time improves over the complexity of a trivial   algorithm that  tries all possible substitutions for the pairs $(?,?)$, and would require a worst time that is   exponential in both the number of species and columns of the input matrix.

\section{An experimental analysis}
\label{experiment}

In  this section we discuss an implementation of  the {\bf  \pp} algorithm. In order to optimize the time complexity, an ad hoc  iterative version of the algorithm has been implemented.
 
We have implemented  and tested the {\bf  \pp} algorithm over simulated data produced by  the tool  \textit{ms} by Hudson \cite{Hu}.   The test set consisted in a random data set of matrices generated with a recombination rate of $1$ over $15$.  The main goal of the experimental analysis has been to test the applicability of the algorithm to matrices with different complexities in terms of size and number of conflicts (i.e. edges)  in the conflict graph.

We have implemented the algorithm   in C++ and the 
experiments have been run on a standard windows workstation with 4 GB of main memory. 

A preliminary experiment has been done to evaluate the performance of the algorithm with respect to specific parameters related 
to the complexity of the input matrix under mutation events.
Table \ref{tab1} reports  the time computation to solve sets of $50$ matrices for each dimension $(50,15)$,  $(100,15)$,  $(200,15)$, and  $(500,15)$ with a recombination rate $1$ over $15$. 
The table has additional entries to specify  the average time to solve a single matrix (calculated as the ratio between the total time and the number of considered matrices), the  number of matrices that do not admit a p-pp tree, the total number of conflicts in the conflict graphs of the matrices of each set, and the average number of conflicts for each matrix of the set.  Notice that conflicts are measured as the number of edges in graph $G_c$.
Each considered matrix has a conflict graph $G_c$  that consists of a single 
non trivial component. The sets contain only matrices that are solved within 5 minutes. Clearly, the number of unsolved matrices increases with the size of the input matrices.

\begin{table*}
\caption{Summary table}
\label{tab1}
\centering
\begin{scriptsize}
\begin{tabular}{|c|c|c|c|c|c|}
\hline
 {\bf nxm} & {\bf  total time in sec. } & {\bf average time in sec.} & {\bf no p-pp} & {\bf total conflicts} & {\bf average conflicts} \\
\hline
     50x15 &         32.323 &    0.646    &     6    &     236   & 4.72   \\
     100x15 &    194.625      &   3.893    &    4     &    175  & 3.5     \\
     200x15 &    43.212     &     0.864   &     3     &     147  & 2.94    \\
     500x15 &    889.433     &    17.789    &    7      &    219 & 4.38      \\
\hline
\end{tabular}
\end{scriptsize}
\end{table*}

Observe that in general the average running time of the algorithm increases with the size of the input matrix but also with the number of conflicts   that are present in the conflict graph.
This last behavior is suggested  by the results reported in Table \ref{tab2} .

\begin{table*}
\caption{Average execution time in seconds to solve  $10$ matrices with a single conflict. }
\label{tab2}
\centering
\begin{scriptsize}
\begin{tabular}{|c|c|c|c|}
\hline
 {\bf 50x15} & {\bf 100x15} & {\bf 200x15} & {\bf 500x15} \\
\hline
0.015&	0.031&	0.047&	0.093 \\
\hline
\end{tabular}
\end{scriptsize}

\end{table*} 

Another experiment has been done with 10 matrices of the same size $50 \times 15$ and different number of edges in the conflict graph. The average time was  $0.015$,     $0.031$ and  $0.051$, respectively for the case of $1, 5$ and $10$ conflicts. 

In order to test the performance of the algorithm for large matrices in terms of number of species we have processed a matrix of size $1000 \times 15$ with a conflict graph having $9 $ conflicts (edges).  It took $35.5$ seconds to find the solution to the matrix.

\section{Conclusions}

In this paper, we formalize the problem of reconstructing a Persistent Perfect Phylogeny over binary values (P-PP problem); the problem generalizes  the classical directed perfect phylogeny by allowing each character to change from $1$ to $0$ at most once in the tree.
Then, we show that  solving the problem P-PP reduces to a graph-reduction problem.   Based on this combinatorial interpretation of the problem of the persistent perfect phylogeny  problem, we give an exact algorithm for the P-PP problem that has a worst time complexity that is exponential in the number of characters,  but polynomial in the number of species. 
An experimental analysis of the implemented algorithm for the P-PP problem  shows the applicability of the model  to incorporate biological complexity due to  back mutation events in the data.

\section{Acknowledgement}
The first author, PB would like to thanks Russel Schwartz for indicating applications of the model. Moreover, the authors would like to thank Gianluca Della Vedova for comments on a preliminary version of the paper and Francesca Scoglio for technical support in the experimental analysis.
All authors are grateful to the anonymous referees for helpful suggestions on the presentation of the paper.

\bibliographystyle{abbrv}
\bibliography{biblio-pph}


\end{document}